\documentclass [11pt] {article}
\usepackage{geometry}
\usepackage[english]{babel}
\usepackage[utf8x]{inputenc}
\usepackage{amsfonts}
\usepackage{mathtools}
\usepackage{amsmath}
\usepackage{amsthm}
\usepackage{pict2e}
\usepackage{xcolor}

\theoremstyle{plain}
\newtheorem{thm}{Theorem}[section]

\newtheorem{prop}[thm]{Proposition}

\theoremstyle{definition}

\theoremstyle{remark}
\newtheorem{oss}{Observation}

\title{Distortions of Robertson-Walker metric in perturbative cosmology and interpretation as dark matter and cosmological constant}
\author{Federico Re$^*$}

\begin{document}

\maketitle

$^*$DiSAT, Insubria University, Via Valleggio 11, Como, Italy; and INFN, via Celoria 16, 20133, Milano, Italy. fre@uninsubria.it

\begin{center}
	\textbf{Abstract}
\end{center}
In the last years, we saw more and more attempts to explain dark matter as a general relativistic effect, at least for some fraction. Following this philosophy, we considered the gravitational distortions due to the inhomogeneous distribution of matter in the universe, which we know from general relativity to be retarded distortions. This provides a magnification effect, since the distortions we feel now depend not on the present matter density, but on the past one, which is greater.  The expansion rate of the universe is perturbed as well, in a not negligible way, despite matter inhomogeneities are small, because of the same magnification effect. The deceleration parameter, which is a way to evaluate the quantity of dark matter in the universe, is perturbed in turn, so that the real amount of dark matter is less than what is usually believed.

\section{Introduction}

At the present days, the search of dark matter didn't give conclusive results. The hot dark matter, assuming it exists, cannot explain local effects as galaxy rotation curves, since its high mean velocity forbids links to local structures \cite{HDM}. MaCHOs give a contribution to dark matter, but only for a small fraction, according to recent estimates \cite{Alcock:1995dm}, \cite{machos}. The search of WIMPs, conjectured as supersymmetric partners or other kinds of particles, produced no results up to now \cite{Craig:2013cxa}. The state of art suggests that the dark matter is less than what we expect from astrophysical phenomena, but, rather, our Cosmological Standard Model should be modified in some way. This possibility is studied, for example, by the MOND theories, but no attempt to modify the present theory of gravitation given rise to a good matching with data, currently \cite{Aguirre:2001xs}, \cite{mond1}, \cite{mond2}. This situation leads to propose a new kind of answer, which doesn't need some form of unobservable matter, neither modifications to Einstein's theory of gravitation, but explores unusual consequences of the usual theory.

We can roughly divide the “dark matter phenomena” in two categories: global dark matter effects, which consist in unexpected values of cosmological parameters (the deceleration parameter of expansion of universe \cite{perlmutter}, \cite{riess}, the deuterium abundance, the power spectrum of CMB anisotropies, etc.), and local dark matter effects, that arise from observations of astronomical objects (the galaxies rotation curves \cite{rotation curve}, the virial of galaxy clusters \cite{clusters}, gravitational lensing \cite{g lensing}, etc.). Both of them are exclusively gravitational phenomena, anomalies of gravitation with respect to what we expect. This doesn't mean necessarily that there exists some kind of invisible matter, generating the observed gravitation. From general relativity, we know that the gravitational potential is the space-time metric. If we find a general relativistic explanation of the unexpected distortion of this tensor, even without a presence of real matter, we would have the explanation of some dark matter effects. After all, the formalism usually used to show the dark matter effects is not truly general relativistic. Especially for local effects, it is usually used the newtonian approximation; for global ones, the usual model of universe expansion is a friedmanian one with matter and curvature assumed homogeneous, which could be an excessive simplification.

In the last twenty years, several lines of research have been opened that seek to study some essentially relativistic effects in cosmology. One of these started form deep considerations about the coordinates used in general relativity \cite{Alba:2006ea}, \cite{Lusanna:2006wx}, \cite{Lusanna:2009fd}, \cite{Lusanna:2010pb}, \cite{Lusanna:2011rm}, \cite{Lusanna:2012kx}. These have consequences both for global effects, due to different time coordinates in use, and for local ones, due to an off diagonal contribution in the metric tensor for the variables we use in astronomical observation, which can modify the galaxies rotation curves and gives some contributions to the dark matter phenomena \cite{Balasin:2006cg}. This was recently confirmed by observations on our galaxy \cite{Crosta:2018var}.

Another line of research explores the backreaction effect, which means the difference between the spatial curvature due to an averaged quantity of matter, and the averaged spatial curvature due to an inhomogeneous quantity of matter \cite{Buchert:2001sa}, \cite{Vigneron:2019dpj}; it doesn't vanish in general since the Einstein Equations are not linear. This effect was proposed to be also an explanation of the cosmological constant \cite{Buchert:2007ik}. Indeed, also our necessity of a cosmological constant is due to exclusively gravitational phenomena. However, if we try to divide the cosmological constant phenomena in global and local effects, similarly to the dark matter ones, we don't find any local effect. Our knowledge of the cosmological constant comes only from the universe expansion and the measures of its acceleration \cite{perlmutter}, \cite{riess}, which are global effects. This encourages attempts of explanation with averaging on large dominions.

A third and last line considers the retarded gravitational potential generated by the inhomogeneous, expanding distribution of matter in the universe \cite{sergio}. Since it is anisotropic, the Birkhoff Theorem doesn't hold and the central point is influenced by far objects. The potentials from far objects are retarded, so depend on the past matter density, which is greater than the actual one. The causal propagation in gravity, and so its retarded potentials, was recently confirmed by the observation of gravitational waves \cite{waves}. This provides a magnification effect on the total metric tensor, obtained as a superposition of all retarded potentials from all past times, which predictably have a singularity at the Big Bang time. However, there is also a reduction of the gravitational potential with the distance. It is necessary a precise calculation to see if it prevails the magnification or the reduction; in the first case, we would have an explanation of dark matter, as a distortion of the tensor metric not due to a proportional presence of matter. It is possible to approximate the retarded potentials and their superposition with linear perturbation theory, since at large scale the matter distribution is almost homogeneous \cite{homog}, so the inhomogeneities can be seen as small perturbations. This "retarded potential framework" can provide also local dark matter effects, since inhomogeneous matter generates inhomogeneous potentials. The anomalies in rotation curves and virials would be due to a statistical maximum of metric distortion around the galaxies and the clusters, which acts as an halo of effective dark matter. The correspondence of these maximums with the galaxies would not be a coincidence: during the formation of structure, the baryonic matter would fall into the gravitational wells, generated by far fields, so we would have automatically the formation of galaxies inside them.

However, the only paper that contains this kind of idea \cite{sergio} develops it in a very partial way. First of all, it doesn't work in a general relativistic context. They use a linearized gravity model on a minkowskian background, where the simulated expansion of the universe is imposed by hand. Only the matter itself moves, in a fixed Minkowski background, and an effective FRW metric is imposed with a "compatibility condition". These concepts are devoid of analogues in general relativity, where rather there should be a background dynamics and a perturbed one, whose difference is interpretable with some amount of fictitious dark matter. Moreover, in the article \cite{sergio} a relevant role is attribute to a particular fractal distribution of matter. In the present article we will prove that,  in a full general relativistic description, what plays role in determining dark matter effects are the retarded effects and generic inhomogeneities, not necessarily of fractal nature.

The paper is structured as follows. In §2 we linearize the Einstein Equations with respect to a Robertson-Walker metric, assumed spatially flat to simplify all calculations; it is performed the formalism of perturbative cosmology, but we are interested in the development of the metric rather than the one of matter structures, so we choose the gauge used typically for gravitational waves. This allows us to obtain the equation for gravitational wave on an expanding universe, in §3, and so implicitly the perturbed metric. In this work, we focus on global dark matter effects, so we take the average of this perturbed metric, obtaining a general algorithm for the amount of fictitious matter due to the inhomogeneity of the real one. Even if the effect we find is from averaging the metric tensor, it is not a backreaction effect, since there is no spatial curvature in any stadium of calculation, so the backreaction of our model would be zero. Our effect is complementary to backreaction. In §4 we solve the wave equations for a particularly simple case, which however is not qualitatively different from the real universe. In §5 we obtain an explicit formula for the apparent matter depending on the fraction of the real matter inhomogeneously distributed, which results to be not negligible.

\medskip
\textbf{Notations}. In this article is used the most minus signature and natural units, so that $c=1$. The Ricci tensor is defined as $R_{\mu\nu}=\partial_{\sigma}\Gamma^{\sigma}_{\mu\nu}-\partial_{\nu}\Gamma^{\sigma}_{\mu\sigma}+\Gamma^{\sigma}_{\sigma\rho}\Gamma^{\rho}_{\mu\nu}-\Gamma^{\sigma}_{\mu\rho}\Gamma^{\rho}_{\nu\sigma}$, with $\Gamma^{\lambda}_{\mu\nu}=\frac{1}{2}g^{\lambda\rho}(\partial_{\mu}g_{\nu\rho}+\partial_{\nu}g_{\rho\mu}-\partial_{\rho}g_{\mu\nu})$.

Any quantity has an "unperturbed", or "averaged", version $\overline{Q}$ calculated from the background metric $\overline{g}_{\mu\nu}(\tau)$, and a "perturbed", or "true", one $Q$ calculated from the real metric $g_{\mu\nu}(\tau;\underline{x})$. Its "perturbation" is the difference $\tilde{Q}:=Q-\overline{Q}$.

$\tau$ is the conformal unperturbed time, with allowed values $(\tau_I;\tau_F):=\{\tau\in\mathbb{R}|a(\tau)>0\}$, where $a(\tau)$ is the unperturbed expansion parameter; so that $\overline{t}=\int a(\tau)d\tau$ is the usual, unperturbed time. $H(\tau)$ is the Hubble parameter for the unperturbed model. $\tau$, $a$ and $H$ are written without the overline, with notation abuse, for a better readability. Their perturbed versions will be $\textbf{a}$, $\textbf{H}$ and so on.

$\tau_0$ is the actual instant in the unperturbed model, s.t. $a(\tau_0)=1$. $t_0$ is the present in the true model, s.t. $\textbf{a}(t_0)=1$. The "$0$" label means an evaluation in the present time.

$\hat{V}$ is the divergenceless part of a vector $V$.

\section{Linearized Einstein Equations}

Let us apply the perturbative approach to a Robertson-Walker metric, looking for a wave equation for the perturbation. As usual \cite{perturb4}, \cite{oliver}, we add a small perturbation to an unperturbed metric:
\begin{equation}
	g_{\mu\nu}(\tau;\underline{x})=\overline{g}_{\mu\nu}(\tau)+\tilde{g}_{\mu\nu}(\tau;\underline{x})=a(\tau)^2\left(\begin{matrix}
		1 & \vec{0} \\
		\vec{0} & -\delta_{ij}
	\end{matrix}\right)+a(\tau)^2\left(\begin{matrix}
		2A & -\vec{B} \\
		-\vec{B} & h_{ij}
	\end{matrix}\right).
\end{equation}
\begin{oss}
\label{back1}
	 In order to keep all calculations as simple as possible, we assume a spatially flat background. As we will see later (Observation \ref{back2}), also the perturbation won't have spatial curvature. On any compact support $\mathcal{D}$, the spatial second fundamental form $K_{ij}$ will have a constant trace $\theta:=-K_{ij}g^{ij}$ and an identically zero shear tensor $\sigma_{ij}:=-K_{ij}-\frac{1}{3}\theta g_{ij}\equiv0$. So the kinematical backreaction, see \cite{Buchert:2001sa}, §2.1 and §3.2, will be always zero:
	
	 $\tilde{\mathcal{Q}}_{\mathcal{D}}:=\frac{2}{3}\langle\left(N\theta-\langle N\theta\rangle_{\mathcal{D}}\right)^2\rangle_{\mathcal{D}}-2\langle N^2\frac{1}{2}\sigma^i_j\sigma^j_i\rangle_{\mathcal{D}}\equiv0$.

	 In any case, the backreaction is a second-order quantity, so it never can be found in a calculation at first order, as the our one is.
\end{oss}

As about the energy-momentum tensor, we set
\begin{equation}
	T_{\mu\nu}(\tau;\underline{x})=(\rho+p)U_{\mu}U_{\nu}-pg_{\mu\nu}=\overline{T}_{\mu\nu}(\tau)+\tilde{T}_{\mu\nu}(\tau;\underline{x})
\end{equation}
where the energy density $\rho=\overline{\rho}+\tilde{\rho}$, the pressure $p=\overline{p}+\tilde{p}$ and the four-velocity field $U_{\mu}=a\delta_{\mu\tau}+\tilde{U}_{\mu}$, in general.

The unperturbed Ricci tensor and the unperturbed Einstein tensor are
\begin{equation}
	\overline{R}_{\mu\nu}=\left(\begin{matrix}
		-3\dot{H} & \vec{0} \\
		\vec{0} & (\dot{H}+2H^2)\delta_{ij}
	\end{matrix}\right), \; \overline{G}_{\mu\nu}=\left(\begin{matrix}
		3H^2 & \vec{0} \\
		\vec{0} & -(2\dot{H}+H^2)\delta_{ij}
	\end{matrix}\right).
\end{equation}
The unperturbed Einstein Equations are nothing but the Friedman equations
\begin{equation}
\label{mean einstein}
	\begin{cases}
		3H^2=8\pi Ga^2\overline{\rho} \\
		\dot{\overline{\rho}}=-3H(\overline{\rho}+\overline{p})
	\end{cases}.
\end{equation}
For any superposition of cosmic components
\begin{center}
	$\overline{\rho}=\sum_w\rho_w$, $\overline{p}=\sum_wp_w=\sum_ww\rho_w$
\end{center}
and choosing as variables
\begin{center}
	$\overline{\Omega}_w(\tau):=\frac{\rho_w(\tau)}{\overline{\rho}(0)}$ s.t. $\overline{\rho}(0)=\frac{3H(0)^2}{8\pi G}$, so $\sum_w\overline{\Omega}_w(\tau_0)=1$
\end{center}
we obtain the usual differential equation for the universe expansion:
\begin{equation}
\label{fried}
	\left(\frac{\dot{a}}{H_0}\right)^2=\sum_w\overline{\Omega}_{w0}a^{1-3w}.
\end{equation}
This (\ref{fried}) is the friedmanian model for a general homogeneous universe, for any choice of kind of components $\{w\}$ and of their relative fraction today $\{\overline{\Omega}_{w0}\}$.

The only component we perturb is matter, for which $w=0$, so $\tilde{\rho}=\tilde{\rho}_0$, $\tilde{p}=0\cdot\tilde{\rho}=0$ and $p=\overline{p}$.

After performing the scalar-vector-tensor decomposition of the metric
\begin{equation}
	\begin{matrix}
		\vec{B}:=\vec{\nabla}B+\hat{B}, \; h_{ij}:=2C\delta_{ij}+2\left(\partial_{ij}-\frac{1}{3}\delta_{ij}\nabla^2\right)E+(\partial_i\hat{E}_j+\partial_j\hat{E}_i)+2\hat{E}_{ij} \\
	s.t. \; \vec{\nabla}\cdot\hat{B}=0, \; \vec{\nabla}\cdot\hat{E}=0, \; \sum_j\partial_j\hat{E}_{ij}=0, \; \sum_j\hat{E}_{jj}=0
	\end{matrix}
\end{equation}
we can express the perturbation of the Ricci tensor as:
\begin{align}
\label{ricci non gauge}
	\tilde{R}_{\mu\nu}&:=R_{\mu\nu}-\overline{R}_{\mu\nu}+o(\tilde{g}_{\mu\nu}), \; s.t. \cr
	&\tilde{R}_{\tau\tau}=(\nabla^2A+\nabla^2\dot{B}+3\ddot{C})+H(3\dot{A}+\nabla^2B+3\dot{C}) \cr
	&\tilde{R}_{\tau j}=\left(-\frac{1}{2}\nabla^2\hat{B}_j+\frac{1}{2}\nabla^2\partial_jB+2\partial_j\dot{C}-\frac{1}{2}\nabla^2\dot{\hat{E}}_j\right)-2H\partial_jA+(\dot{H}+2H^2)B_j \cr
	&\tilde{R}_{ij}=\left(\frac{1}{2}\Box h_{ij}-\partial_{ij}A-\partial_{(i}\dot{\hat{B}}_{j)}+\partial_{ij}C-\nabla^2\partial_{(i}\hat{E}_{j)}\right)+\cr
	&-H(\dot{A}\delta_{ij}+\nabla^2B\delta_{ij}+2\partial_{(i}\hat{B}_{j)}+3\dot{C}\delta_{ij}+\dot{h}_{ij})-(\dot{H}+2H^2)(2A\delta_{ij}+h_{ij})
\end{align}
where we call $\partial_{(i}v_{j)}:=\frac{1}{2}(\partial_iv_j+\partial_jv_i)$.

We can use the geometric condition $g_{\mu\nu}U^{\mu}U^{\nu}=1=\overline{g}_{\mu\nu}\overline{U}^{\mu}\overline{U}^{\nu}$, to get the perturbation of velocities as
\begin{equation}
	\tilde{U}^{\mu}=a^{-1}\left(\begin{matrix}
		-A \\
		\vec{v}
	\end{matrix}\right), \; \tilde{U}_{\mu}=a\left(\begin{matrix}
		A & -\vec{v}-\vec{B}
	\end{matrix}\right)
\end{equation}
and so, remembering $\tilde{p}=0$, the perturbation of stress-energy tensor is
\begin{equation}
	\tilde{T}_{\mu\nu}=a^2\left(\begin{matrix}
		\tilde{\rho}+2\overline{\rho}A & -\vec{q}-\overline{\rho}\vec{B} \\
		-\vec{q}-\overline{\rho}\vec{B} & -p h_{ij}
	\end{matrix}\right),
\end{equation}
where $\vec{v}(\tau;\underline{x})$ is the field of spatial velocities, and we defined $\vec{q}:=(\overline{\rho}+p)\vec{v}$.

We want now to deduce the equations for the retarded potentials. To this end, we fix the harmonic gauge, usually convenient for studying gravitational waves. Abstracting from the background metric, the harmonic condition on the perturbation of connection is
\begin{equation}
	\tilde{\Gamma}^{\lambda\mu}_{\mu}=0.
\end{equation}
We obtain a scalar and a vector condition on $A$, $\vec{B}$, $h_{ij}$:
\begin{equation}
	\begin{cases}
		\dot{A}+\nabla^2B+3\dot{C}+4HA=0 \\
		\vec{\nabla}A+\dot{\vec{B}}-\vec{\nabla}C+\nabla^2\vec{E}+2H\vec{B}=0
	\end{cases}.
\end{equation}
In this gauge, the second order part of $\tilde{R}_{\mu\nu}$ is a flat d'alambertian $\Box:=\eta_{\mu\nu}\partial^{\mu}\partial^{\nu}$. Indeed, from (\ref{ricci non gauge}):
\begin{equation}
	\tilde{R}_{\mu\nu}=\left[\frac{1}{2}\Box-H\partial_{\tau}-2(\dot{H}+H^2)\right]\left(\begin{matrix}
		2A & -\vec{B} \\
		-\vec{B} & h_{ij}
	\end{matrix}\right)+\dot{H}\left(\begin{matrix}
		0 & \vec{0} \\
		\vec{0} & h_{ij}-2A\delta_{ij}
	\end{matrix}\right).
\end{equation}
This is what we are looking for, because predictably the linearized Einstein Equations will have the form of wave equations.

Following again the gravitational waves formalism, we express the Einstein Field Equations as
\begin{equation}
	R_{\mu\nu}=8\pi GS_{\mu\nu}=8\pi G\left(T_{\mu\nu}-\frac{1}{2}T^{\lambda}_{\lambda}g_{\mu\nu}\right).
\end{equation}
The perturbation of the $S$ tensor is
\begin{equation}
	\tilde{S}_{\mu\nu}=\frac{a^2}{2}\left(\begin{matrix}
		\tilde{\rho}+2(\overline{\rho}+3p)A & -2\vec{q}-(\overline{\rho}+3p)\vec{B} \\
		-2\vec{q}-(\overline{\rho}+3p)\vec{B} & \tilde{\rho}\delta_{ij}+(p-\overline{\rho})h_{ij}
	\end{matrix}\right)
\end{equation}
and so the linearized Einstein Equations, simplified using equations (\ref{mean einstein}), are
\begin{equation}
\label{sys}
	\begin{cases}
		\Box A-2H\dot{A}+2(\dot{H}-2H^2)A=4\pi Ga^2\tilde{\rho} \\
		\Box\vec{B}-2H\dot{\vec{B}}+2(\dot{H}-2H^2)\vec{B}=16\pi Ga^2\vec{q} \\
		\Box h_{ij}-2H\dot{h}_{ij}=4(\dot{H}A+2\pi Ga^2\tilde{\rho})\delta_{ij}
	\end{cases}.
\end{equation}
From them, we can derive the linearized conservation laws for energy and momentum:
\begin{equation}
	\begin{cases}
		\dot{\tilde{\rho}}+\vec{\nabla}\cdot\vec{q}+3H\tilde{\rho}-3(\overline{\rho}+p)\dot{C}=0 \\
		\dot{\vec{q}}+4H\vec{q}+(\overline{\rho}+p)(\vec{\nabla}C+\nabla^2\hat{E})+[(\dot{\overline{\rho}}+\dot{p})+2H(\overline{\rho}+p)]\vec{B}=0
	\end{cases}.
\end{equation}

\section{Retarded potentials and averaged perturbed metric}

The solutions in the vacuum of the PDE system (\ref{sys}) describe gravitational waves on an expanding space-time. For a given distribution of matter and velocities as source, the PDEs return the correspondent space-time metric. For a bounded distribution of matter, the solution without gravitational waves is such that the metric is asymptotically minkowskian, and we choose this solution as gravitational potential.

Similarly to the usual wave equation, the characteristic curves are light rays, and so the potentials are retarded accordingly to the speed of light.

\subsection{Simplification of wave equations and general form of retarded potentials}

First, we observe that $h_{ij}$ has no traceless source, hence we can choose a solution with $h_{ij}=2C\delta_{ij}$.

Moreover, let us decompose $\vec{q}:=\vec{\nabla}q+\hat{q}$. The divergenceless part $\hat{B}$ has only $\hat{q}$ as source and both of them are decoupled from the rest of the system\footnote{In the energy conservation, it appears only the divergence of $\vec{q}$. Since $h_{ij}$ is now a pure trace, the momentum conservation is $\dot{\vec{q}}+4H\vec{q}+(\overline{\rho}+p)(\vec{\nabla}C)+[(\dot{\overline{\rho}}+\dot{p})+2H(\overline{\rho}+p)]\vec{B}=0$. Decomposing it in a pure gradient and a divergenceless part, the dependence on $C$ survives only in the first one, while the second one links only $\hat{B}$ and $\hat{q}$, again.}. We choose to set both to zero, which means to consider an irrotational dust as inhomogeneous matter.

The system becomes
\begin{equation}
	\begin{cases}
		\Box A-2H\dot{A}+2(\dot{H}-2H^2)A=4\pi Ga^2\tilde{\rho} \\
		\Box B-2H\dot{B}+2(\dot{H}-2H^2)B=16\pi Ga^2q \\
		\Box C-2H\dot{C}=2\dot{H}A+4\pi Ga^2\tilde{\rho}
	\end{cases}.
\end{equation}
\begin{oss}
\label{back2}
	Let us consider now the total metric $g_{\mu\nu}=a^2\left(\begin{matrix}
		1+2A & -\vec{\nabla}B \\
		-\vec{\nabla}B & (2C-1)\delta_{ij}
	\end{matrix}\right)$. As we said in Observation \ref{back1}, the spatial part is flat, even in perturbation, so we won't have backreaction.
\end{oss}
All the wave equations have the PDE form
\begin{equation}
\label{PDE}
	\Box u+\mathcal{H}(\tau)\dot{u}+\mathcal{K}(\tau)u=\mathcal{S}(\tau;\underline{x}).
\end{equation}
Let $G(\tau, \underline{x}; \tau', \underline{x}')$ be its Green function\footnote{i.e. the asymptotically zero solution for a source $\delta(\tau-\tau')\delta^{(3)}(\underline{x}-\underline{x}')$}. It will be zero for $|\underline{x}-\underline{x}'|>\tau-\tau'$, because of causality. It is also spatially homogeneous and isotropic:
\begin{center}
	$G(\tau, \underline{x}; \tau', \underline{x}')=G(\tau, \underline{x}-\underline{x}'; \tau', \underline{0})=G(\tau, |\underline{x}-\underline{x}'|; \tau', \underline{0})$.
\end{center}
Assuming separation of variables for a generic source $\mathcal{S}(\tau; \underline{x})=T(\tau)\mathcal{S}_0(\underline{x})$, we can express the retarded potential as
\begin{equation}
	\begin{matrix}
		u(\tau;\underline{x})=\int_{\tau_I}^{\tau}d\tau'\int d^3\underline{x}' G(\tau, \underline{x}; \tau', \underline{x}')T(\tau')\mathcal{S}_0(\underline{x}')=\int_{|\underline{x}'-\underline{x}|<\tau-\tau_I}\mathcal{S}_0(\underline{x}')f(\tau; |\underline{x}'-\underline{x}|)d^3\underline{x}' \\
	s.t. \; f(\tau; |\underline{r}|):=\int_{\tau_I}^{\tau}G(\tau, \underline{r}; \tau', \underline{0})T(\tau')d\tau'.
	\end{matrix}
\end{equation}

The $f$ auxiliary quantity shows the superposition of all the retarded potentials generated by a point of the source at all times, from the Big Bang up to now. The resultant solution $u$ is again the superposition for all the causally linked points.

\subsection{Averaging the perturbed metric}

To study the global effects, we take the average of these potentials over all the space, obtaining a function depending only on time. We get
\begin{prop}
	\begin{equation}
	\label{average formula}
		\langle u\rangle(\tau)=4\pi\langle\mathcal{S}_0\rangle\int_0^{R(\tau)}f(\tau; r)r^2dr,
	\end{equation}
	where $R(\tau):=\tau-\tau_I$ is the radius of observable universe.
\end{prop}
\begin{proof}
	The average of a spatial quantity $Q(\underline{x})$ on some compact support $\mathcal{D}$ is defined as

	$\langle Q\rangle_{\mathcal{D}}:=\frac{1}{|\mathcal{D}|}\int_{\mathcal{D}}Q(\underline{x})d^3\underline{x}$.

	The average on the whole space is the limit of this quantity for a monotonically growing sequence of sets  $\mathcal{D}$, tending to $\mathbb{R}^3$:

	$\langle Q\rangle:=\lim_{\mathcal{D}\nearrow\mathbb{R}^3}\langle Q\rangle_{\mathcal{D}}$.

	Let us fix the time $\tau$. What we obtain immediately for $u$ is

	$\langle u\rangle(\tau):=\lim_{\mathcal{D}\nearrow\mathbb{R}^3}\frac{1}{|\mathcal{D}|}\int_{\mathcal{D}}u(\tau;\underline{x})d^3\underline{x}=$

	$=\lim_{\mathcal{D}\nearrow\mathbb{R}^3}\frac{1}{|\mathcal{D}|}\int_{\mathcal{D}}d^3\underline{x}\int_{|\underline{x}'-\underline{x}|<R(\tau)}d^3\underline{x}'\mathcal{S}_0(\underline{x}')f(\tau;|\underline{x}-\underline{x}'|)=$
	
	$=\lim_{\mathcal{D}\nearrow\mathbb{R}^3}\frac{1}{|\mathcal{D}|}\int_{\mathcal{D}}d^3\underline{x}\int_{|\underline{r}|<R(\tau)}d^3\underline{r}\mathcal{S}_0(\underline{r}+\underline{x})f(\tau;|\underline{r}|)=$

	$=\int_{|\underline{r}|<R(\tau)}f(\tau;|\underline{r}|)\left[\lim_{\mathcal{D}\nearrow\mathbb{R}^3}\frac{1}{|\mathcal{D}|}\int_{\mathcal{D}}d^3\underline{x}\mathcal{S}_0(\underline{r}+\underline{x})\right]d^3\underline{r}=$

	$=\int_{|\underline{r}|<R(\tau)}f(\tau;|\underline{r}|)\left[\lim_{\mathcal{D}+\underline{r}\nearrow\mathbb{R}^3}\langle\mathcal{S}_0\rangle_{\mathcal{D}+\underline{r}}\right]d^3\underline{r}=$

	$=\int_{r<R(\tau)}f(\tau;r)\langle\mathcal{S}_0\rangle 4\pi r^2dr$,

	which\footnote{We call $\mathcal{D}+\underline{r}:=\{\underline{x}+\underline{r}|\underline{x}\in\mathcal{D}\}$ in the fourth passage, and we changed the variables to polar in the fifth passage.} proves the proposition.
\end{proof}
For $B$ it is necessary a different procedure, since it is not a component of the metric, but their partial derivatives $\vec{\nabla}B$ are. Its average results to be zero:
\begin{prop}
	\begin{equation}
		\langle\vec{\nabla}B\rangle=0.
	\end{equation}
\end{prop}
\begin{proof}
	Similarly to $u$, we find now the average of $\vec{\nabla}u$:

	$\langle \vec{\nabla}u\rangle(\tau):=\lim_{\mathcal{D}\nearrow\mathbb{R}^3}\frac{1}{|\mathcal{D}|}\int_{\mathcal{D}}\vec{\nabla}u(\tau;\underline{x})d^3\underline{x}=$

	$=\lim_{\mathcal{D}\nearrow\mathbb{R}^3}\frac{1}{|\mathcal{D}|}\int_{\mathcal{D}}d^3\underline{x}\int_{|\underline{x}'-\underline{x}|<R(\tau)}d^3\underline{x}'\mathcal{S}_0(\underline{x}')\vec{\nabla}_{\underline{x}}f(\tau;|\underline{x}-\underline{x}'|)=$
	
	$=\lim_{\mathcal{D}\nearrow\mathbb{R}^3}\frac{1}{|\mathcal{D}|}\int_{\mathcal{D}}d^3\underline{x}\int_{|\underline{r}|<R(\tau)}d^3\underline{r}\mathcal{S}_0(\underline{r}+\underline{x})\vec{\nabla}_{\underline{r}}f(\tau;|\underline{r}|)=$

	$=\int_{|\underline{r}|<R(\tau)}\left(\vec{\nabla}_{\underline{r}}f(\tau;|\underline{r}|)\right)\left[\lim_{\mathcal{D}\nearrow\mathbb{R}^3}\frac{1}{|\mathcal{D}|}\int_{\mathcal{D}}d^3\underline{x}\mathcal{S}_0(\underline{r}+\underline{x})\right]d^3\underline{r}=$

	$=\int_{|\underline{r}|<R(\tau)}\frac{\underline{r}}{|\underline{r}|}f'(\tau;|\underline{r}|)\langle\mathcal{S}_0\rangle d^3\underline{r}=0$,

	because it is an integral of an odd function over a symmetric region.In particular, it is true for $u=B$.
\end{proof}
The averaged metric we find is diagonal:
\begin{equation}
	\langle g_{\mu\nu}\rangle=\overline{g}_{\mu\nu}+\langle\tilde{g}_{\mu\nu}\rangle=a(\tau)^2(1+2\langle A\rangle(\tau))d\tau\otimes d\tau-a(\tau)^2(1-2\langle C\rangle(\tau))\delta_{ij}dx_idx_j,
\end{equation}
where
\begin{multline}
	\langle A\rangle(\tau)=16\pi^2G\langle\tilde{\rho}\rangle(\tau_0)\int_0^{R(\tau)}f_A(\tau; r)r^2dr, \\
	\langle C\rangle(\tau)=4\pi(2\dot{H}_0\langle A\rangle(\tau_0)+4\pi G\langle\tilde{\rho}\rangle(\tau_0))\int_0^{R(\tau)}f_C(\tau; r)r^2dr, \; s.t. \\
	f_{A,C}(\tau; |\underline{r}|):=\int_{\tau_I}^{\tau}G_{A,C}(\tau, \underline{r}; \tau', \underline{0})T(\tau')d\tau', \; T(\tau)=\frac{\tilde{\rho}(\tau,\underline{x})}{\tilde{\rho}(\tau_0,\underline{x})}, \\
	\left(\Box-2H\partial_{\tau}+2(\dot{H}-2H^2)\right)G_A(\tau, \underline{x}; \tau', \underline{x}')=\delta(\tau-\tau')\delta^{(3)}(\underline{x}-\underline{x}'), \\
	\left(\Box-2H\partial_{\tau}\right)G_C(\tau, \underline{x}; \tau', \underline{x}')=\delta(\tau-\tau')\delta^{(3)}(\underline{x}-\underline{x}').
\end{multline}
\begin{oss}
\label{decup}
	$A$ and $C$ are not decoupled, in general. The source $2\dot{H}A$ of $C$ has not separable variables, as we assumed. This is a limit for the averaging procedure, but in the following particular case we will consider the wave equations will decouple.
\end{oss}

\subsection{Comparison with a pure homogeneous model}

Let us consider now an observer inside the universe we are describing. He can observe his universe and try to explain his observations with general relativity. Following a friedmanian paradigm, he could approximate the components inside the universe and its metric as homogeneous. With the usual Robertson-Walker variables, he writes the metric as
\begin{equation}
\label{pert metric}
	\begin{matrix}
		\langle g_{\mu\nu}\rangle=dt\otimes dt-\bold{a}(t)\delta_{ij}dx_idx_j, \; s.t. \\
		\; dt:=a(\tau)\sqrt{1+2\langle A\rangle(\tau)}d\tau, \; \bold{a}:=a(\tau)\sqrt{1-2\langle C\rangle(\tau)}.
	\end{matrix}
\end{equation}
He has a perfect correspondence with the averaged perturbed metric, considering his time $t$ and expansion parameter $\bold{a}(t)$ as perturbed quantities. We can express the correspondence with suitable "perturbations":
\begin{equation}
	\begin{matrix}
		dt:=\tilde{t}d\overline{t} \; s.t. \; \tilde{t}:=\sqrt{1+2\langle A\rangle}, \\
		\bold{a}:=a\tilde{a} \; s.t. \; \tilde{a}:=\sqrt{1-2\langle C\rangle}.
	\end{matrix}
\end{equation}
\begin{oss}
\label{cut off}
	We can have singularities if $\tilde{t}$ or $\tilde{a}$ reach zero. It could happen for times ancient enough, and our perturbation theory is no more valid for previous instants, since imaginary quantities are not allowed.

	If $\langle C\rangle(t_{BB})=\frac{1}{2}$, then $\bold{a}(t_{BB})=0$ even if $a(t_{BB})\neq0$. $t_{BB}$ would be a perturbed Big Bang, and we can set $t_{BB}=0$ w.l.o.g.

	If $\langle A\rangle(t_{min})=-\frac{1}{2}$, then the perturbation theory loses validity even if there is no Big Bang. In this case, too much early epochs remain simply not describable by the model. This provides a cut off for the time integration.
\end{oss}
The equation of expansion that the observer tries to use is the usual friedmanian one (\ref{fried}), interpreting the true expansion $\textbf{a}(t)$ as the result of some effective, homogeneous components $\{\Omega_{w0}\}$. These are not, in general, the unperturbed mean quantities $\{\overline{\Omega}_{w0}\}$. Moreover, the true function $\textbf{a}(t)$ doesn't follow a ODE as (\ref{fried}), in general, so the expansion parameter our observer obtains is not the same of $\textbf{a}(t)$. We can call $a_D(t)$ the parameter of the observer's model, with $H_D$ Hubble parameter, so the observer's ODE is
\begin{equation}
\label{dark fried}
	\left(\frac{\dot{a}_D}{H_{D0}}\right)^2=\sum_w\Omega_{w0}a_D^{-3w-1}.
\end{equation}
\begin{oss}
	Here we call $H_{D0}:=H_D(t_0)$, where $t_0$ is such that $\textbf{a}(t_0)=1$, so in general $t_0\neq t(\tau_0)$ because $\tilde{a}(t_0)\neq1$.
\end{oss}
Trying to evaluate the $\{\Omega_{w0}\}$, the observer can evaluate the quantity of matter he can see, or that of which he can deduce the existence\footnote{excluding gravitational phenomena}:

$\Omega_{BM0}:=\frac{8\pi G}{3H_{D0}^2}\rho_{BM0}$.

In the actual cosmology, $\rho_{BM0}$ is usually called baryonic matter. In general, it is a fraction of the true total matter\footnote{Here we call $\overline{\rho}_{M0}:=\overline{\rho}_{w0}|_{w=0}$ the matter component, for which $w=0$, as after $\Omega_{M0}:=\Omega_{w0}|_{w=0}$.}:

$\rho_{BM0}\leq\overline{\rho}_{M0}+\tilde{\rho}$.

What remains can be called "true dark matter", which actually can exist but is beyond the scientific and technological capabilities of the instruments available to the observer:

$\rho_{TDM0}:=\overline{\rho}_{M0}+\tilde{\rho}-\rho_{BM0}$.

Fitting $a_D$ to the expansion he actually observes, he will need in general an effective matter which is not the matter he sees, nor the total matter really exists. He will call "dark matter" all the lack w.r.t. his evaluating

$\Omega_{DM0}:=\Omega_{M0}-\Omega_{BM0}$,

but not all of this is actually existing matter. There is an amount of "fictitious matter"

$\Omega_{FM0}:=\Omega_{DM0}-\Omega_{TDM0}=\Omega_{M0}-\frac{8\pi G}{3H_{D0}^2}(\overline{\rho}_{M0}+\tilde{\rho})$

which is the quantity we want to estimate now. This $\Omega_{FM0}$ evaluates the difference between the amount of dark matter predicted by usual Cosmological Model, $\Omega_{DM0}$, and the real quantity of dark matter in the universe, $\Omega_{TDM0}$, which is less.

How can the observer estimate $\{\Omega_{w0}\}$? He can measure the expansion rate and its acceleration at the present instant $\dot{\bold{a}}(t_0)$, $\ddot{\bold{a}}(t_0)$, set them equal to his model parameters $\dot{a}_D(t_0)$, $\ddot{a}_D(t_0)$ and obtain a condition from the derivative of (\ref{dark fried}).
\begin{equation}
\label{dark sys}
	\begin{cases}
		a_D(t_0):=1=\bold{a}(t_0) \\
		\dot{a}_D(t_0)=H_{D0}:=\dot{\bold{a}}(t_0) \\
		\ddot{a}_D(t_0):=\ddot{\bold{a}}(t_0)
	\end{cases}.
\end{equation}
This was the procedure followed in \cite{perlmutter}, \cite{riess}, where it is obtained the quantity of matter and cosmological constant. The observer can fit the same parameters. As $\Omega_{M0}$ can contain some amount of fictitious matter $\Omega_{FM0}$, it is possible that also $\Omega_{\Lambda0}:=\Omega_{w0}|_{w=-1}$ would have a fictitious component

$\Omega_{F\Lambda0}:=\Omega_{\Lambda0}-\frac{8\pi G}{3H_{D0}^2}\overline{\rho}_{\Lambda0}$.

We can assume the observer can measure the quantity of any other component, setting

$\Omega_{w0}=\frac{8\pi G}{3H_{D0}^2}\overline{\rho}_{w0}, \forall w\neq0 \; , \; 1$.

In (\ref{dark sys}) the first two equations fix the parameters $t_0, H_0$, which are invisible for the observer, who can measure only the true Hubble parameter $\textbf{H}_0=H_{D0}$. The only unknown parameters remain $\Omega_{DM0}, \Omega_{FM0}, \Omega_{F\Lambda0}$. The third equation provides the algebraic relation
\begin{equation}
\label{alg rel}
	\sum_w(3w+1)\Omega_{w0}=-\frac{2\ddot{\textbf{a}}(t_0)}{H_{D0}^2}.
\end{equation}
Another sure relation is that the sum of all $\Omega_{w0}$ is $1$. So, only one unknown parameter remains, maybe $\Omega_{TDM0}$, which measures the technological observation skill. If the observer manages to evaluate all the actually existing matter in his universe, also $\Omega_{TDM0}$ is fixed at zero and we can calculate all the parameters.
\begin{oss}
\label{magn}
	The magnification effect we talked about in the introduction is verified if $\Omega_{M0}>\frac{8\pi G}{3H_{D0}^2}(\overline{\rho}_{M0}+\tilde{\rho})$. In this case, the quantity of dark matter can be reduced by a quantity $\Omega_{FM0}>0$. If the other case, it means that the reduction effect results to be stronger.
\end{oss}

\section{Constant coefficients case: solution, matter source and effective density}

To obtain some explicit solution of (\ref{PDE}), at least for a simple case, from now on we consider a universe with constant Hubble parameter, so that
\begin{equation}
	\mathcal{H}\equiv\mathcal{H}_0:=-2H_0; \; \mathcal{K}\equiv\mathcal{K}_0:=-4H_0^2 \; for \; A \; and \; B, \; \mathcal{K}\equiv0 \; for \; C.
\end{equation}
In this case the system is completely decupled, so the problem in Observation \ref{decup} is overcome. Physically, we have a universe dominated by only one component $\Omega_{w0}=1$, which is a suitable form of dark energy\footnote{Any component with $w<0$ is a dark energy, while the cosmological constant is a particular form of dark energy with $w=-1$.} such that $w=-\frac{1}{3}$. It is the same expansion assumed in \cite{sergio}, since the background expansion law is $a(\tau)=e^{H_0\tau}$ and $a(\overline{t})=H_0\overline{t}$, but now it is treated in a general relativistic context.

Through Fourier transformation, we obtain the Green function for (\ref{PDE}) with constant coefficients:
\begin{equation}
\label{Green}
	G(\tau;\underline{x})=\frac{e^{\frac{1}{2}\mathcal{H}_0\tau}}{4\pi}\left[-\frac{\delta(\tau-|\underline{x}|)}{|\underline{x}|}+\sqrt{\frac{\overline{\mathcal{K}}}{\tau^2-|\underline{x}|^2}}J_0'\left(\sqrt{\overline{\mathcal{K}}(\tau^2-|\underline{x}|^2)}\right)\theta(\tau-|\underline{x}|)\right],
\end{equation}
where we defined the \emph{discriminant} $\overline{\mathcal{K}}:=-\mathcal{K}_0-\left(\frac{\mathcal{H}_0}{2}\right)^2$, $J_0'$ is the first derivative of the zeroth order Bessel function, and $\theta$ is the Heaviside function. Notice that for $A$ and $B$ we have $\overline{\mathcal{K}}=3H_0^2$, whence it comes a factor $\sqrt{3}$, while for $C$ it is $\overline{\mathcal{K}}=-H_0^2$, for which $J_0$ is replaced by $I_0$, the zeroth order modified Bessel function. We easily recognize the causality in the potential, since the first term propagates at the speed of light and the second one slower. The second term is some kind of "echo", due to the difference of the differential operator from a pure d'alembertian. In this case the PDE is homogeneous in time, so $G(\tau,\underline{x};\tau',\underline{0})=G(\tau-\tau',\underline{x};0,\underline{0})=G(\tau-\tau';\underline{x})$.

\subsection{Is the constant coefficients case representative for the real universe dynamics?}

The answer to the question could not be clear, since a constant expansion is pretty different from the our real universe's one. From the next subsection onwards, we will calculate the fictitious matter in a constant coefficients universe, but even if we find some, there could be doubts about the presence of the same effect in a universe with not constant expansion. The general solution for the wave equations is quite difficult to get and needs numerical integration, but here we now show that it would lead to the same effect. This is because the Green functions have the same shape in any case, inducing similar averaged metric and similar distortion on $\textbf{a}(t)$.

The general PDE (\ref{PDE}) is isotropic, which allows us to reduce the dimension of the problem:
\begin{align}
	G(\tau,\underline{x};\tau',\underline{0})&=-\frac{a(\tau')}{2\pi|\underline{x}|a(\tau)}\partial_r\Phi(\tau,|\underline{x}|;\tau') \; s.t. \cr
	\Phi''-\ddot{\Phi}-\overline{\mathcal{K}}(\tau)\Phi&=\delta(r)\delta(\tau-\tau'); \; \overline{\mathcal{K}}:=\frac{1}{2}\dot{\mathcal{H}}-\frac{1}{4}\mathcal{H}^2-\mathcal{K}.
\end{align}
For example, with constant coefficients $\overline{\mathcal{K}}$ is constant and we derive (\ref{Green}) from

$\Phi(\tau,r;0)=\frac{1}{2}\theta(\tau-r)J_0\left(\sqrt{\overline{\mathcal{K}}(\tau^2-r^2)}\right)$.

As we saw in Observation \ref{decup}, the averaging procedure is not applicable as long as $A, C$ are coupled. It is possible to decouple their PDEs whenever the universe is dominated by a single component $\overline{\Omega}_{w0}=1$. Indeed, we can define an auxiliary field $D$ with equation
\begin{equation}
	\Box D-2H\dot{D}=8\pi G(1-\alpha)a^2\tilde{\rho} \quad s.t. \quad \alpha:=\frac{2}{3w+1}
\end{equation}
and obtain $C$ and its average as
\begin{equation}
	C=\frac{D-A}{1-2\alpha}.
\end{equation}
For a general background, with more components, it is possible to approximate the development of $A, C$ neglecting at each instant all the components except the biggest one. The obtained law is a gluing of more single-component developments.

With a simple numerical integration, we obtained the one-dimensional retarded potentials $\Phi(\tau,r;\tau')$ for some single-component cases. Almost all of them have the same shape, as we see from Figure 1,2,3\footnote{The thickness of graphics denotes the uncertainty of numerical algorithm.}, which allow us to believe that in general case a similar apparent matter will arise. Probably, the numerical value of $\Omega_{FM0}$ is different for our real universe, but we can expect qualitatively the same effect.

The only difference is for the cosmological constant's dominance; in this case, the gravitational wave equation has a pure d'alembertian, so we have no the "echo" term. However, even if in our universe there is a cosmological constant, it is not dominant until very recent times.
\begin{figure}[ht]
	\begin{minipage}[b]{0.5\linewidth}
		\centering
		\includegraphics[width=\textwidth]{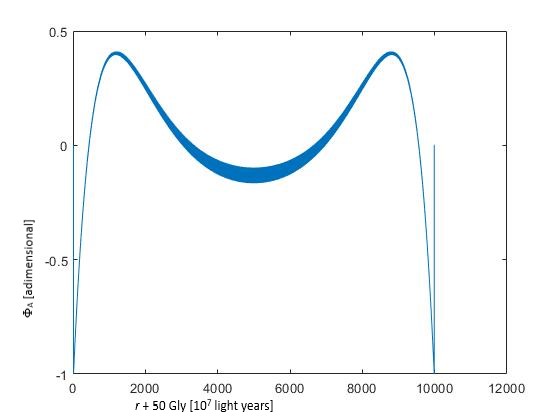}
		\caption{$\Phi_A(\tau_0,r;\tau_0-50 Gy)$ for the dominance of $w=-\frac{1}{3}$ dark energy-kind}
		\label{fig:figure1}
	\end{minipage}
	\hspace{0.2cm}
	\begin{minipage}[b]{0.5\linewidth}
		\centering
		\includegraphics[width=\textwidth]{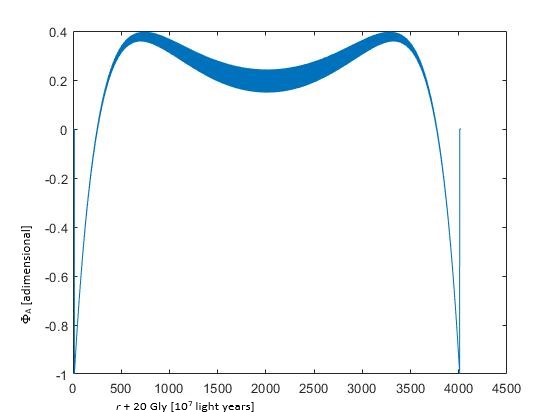}
		\caption{$\Phi_A(\tau_0,r;\tau_0-20 Gy)$ for the dominance of matter}
		\label{fig:figure2}
	\end{minipage}
\end{figure}

\begin{figure}[ht]
	\begin{minipage}[b]{0.5\linewidth}
		\centering
		\includegraphics[width=\textwidth]{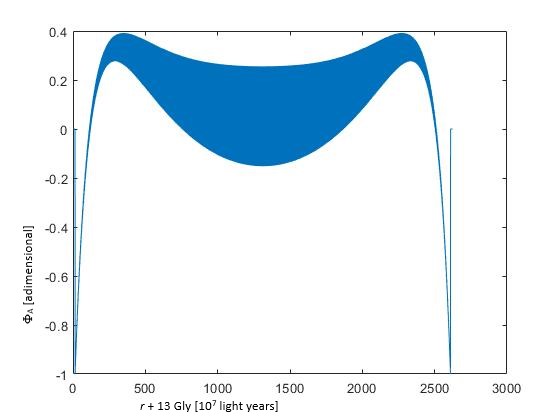}
		\caption{$\Phi_A(\tau_0,r;\tau_0-13 Gy)$ for the dominance of radiation}
		\label{fig:figure1}
	\end{minipage}
	\hspace{0.2cm}
	\begin{minipage}[b]{0.5\linewidth}
		
	\end{minipage}
\end{figure}
\newpage

\subsection{Density contrast growing rate}

To obtain an explicit expression for $\langle A,C\rangle$, what we need now  is an evaluation of their source $\mathcal{S}(\tau;\underline{x})=T(\tau)\mathcal{S}_0(\underline{x})$. Since $\mathcal{S}_0\propto\tilde{\rho}_0$, it is enough to compute the growing law $T(\tau)$ of the matter inhomogeneities, due to the progressive attraction of more and more material from the medium $\overline{\rho}$. Let us measure it with the density contrast of matter
\begin{equation}
	\delta_M:=\frac{\tilde{\rho}}{\overline{\rho}_M}\propto\tilde{\rho}a^3.
\end{equation}
We can obtain a condition on it from the wave equation for $B$ and the energy-momentum conservation, which in the case $w=-\frac{1}{3}$ become
\begin{equation}
	\begin{cases}
		\Box B-2H_0\dot{B}-4H_0^2B=16\pi Ga^2q \\
		\dot{\tilde{\rho}}+3H_0\tilde{\rho}+\nabla^2q-2\overline{\rho}\dot{C}=0 \\
		\dot{q}+4H_0q+\frac{2}{3}\overline{\rho}C=0.
	\end{cases}
\end{equation}
With suitable substitutions,
\begin{equation}
\label{dens}
	\begin{matrix}
		H_0^2p_2(H_0^{-1}\partial_{\tau})\nabla^2\delta_M=H_0^4p_4(H_0^{-1}\partial_{\tau})\delta_M, \; s.t. \; the \; polynomials \; are \\
		\; p_2(x):=x^2+x-1, \; p_4(x):=(x-1)(x+1)(x^2+x+3).
	\end{matrix}
\end{equation}
Since the source has separable variables, we can separate them also for $\delta_M$. We assume now a growing law for the density contrast such that
\begin{equation}
	\delta_M(\tau;\underline{x})=a(\tau)^nX(H_0\underline{x})
\end{equation}
where the parameter $n$ is a "growing rate". From (\ref{dens}), we know
\begin{equation}
	\nabla^2X(\underline{\xi})=\frac{p_4(n)}{p_2(n)}X(\underline{\xi}).
\end{equation}
This means that $X$ is an eigenfunction of the laplacian, so it must have negative eigenvalues. Substituting inside $p_4$ and $p_2$, this condition on $n$ becomes
\begin{equation}
\label{cond n}
	n\in(-\Phi; -1)\sqcup(\varphi; 1)
\end{equation}
where we call $\varphi=\Phi^{-1}:=\frac{\sqrt{5}-1}{2}\cong0.618\dots$ the golden ratio. $\delta_M$ must grow, for gravity, so the physically acceptable values for $n$ are $(\varphi; 1)$. We will use
\begin{equation}
	\tilde{\rho}(\tau;\underline{x})=a(\tau)^{n-3}\tilde{\rho}_0(\underline{x}) \; \Rightarrow \; T(\tau)=a(\tau)^{n-1}.
\end{equation}

\subsection{A formula for the effective density}

Now we can evaluate the average metric with formula (\ref{average formula}), remembering $R(\tau)\equiv+\infty$ in our case. For $A$:
\begin{center}
	$f(\tau; r)=Ga(\tau)^{n-1}\left[-\frac{e^{-nH_0r}}{r}+\sqrt{3}H_0\int_0^{\infty}e^{-nH_0(\sigma+r)}\frac{J_0'(\sqrt{3}H_0\sqrt{\sigma(\sigma+2r)})}{\sqrt{\sigma(\sigma+2r)}}d\sigma\right]$,
\end{center}
\begin{equation}
	\langle A\rangle(\tau)=4\pi\left(\frac{1}{3}\mathcal{N}(n)-\frac{1}{n^2}\right)\frac{G\langle\tilde{\rho}_0\rangle}{H_0^2}a(\tau)^{n-1}.
\end{equation}
Similarly, for $C$:
\begin{equation}
	\langle C\rangle(\tau)=4\pi\left(\mathcal{M}(n)-\frac{1}{n^2}\right)\frac{G\langle\tilde{\rho}_0\rangle}{H_0^2}a(\tau)^{n-1}.
\end{equation}
Here we defined the integrals
\begin{equation}
	\begin{matrix}
		\mathcal{N}(n):=\int_0^{\infty}\int_0^{\infty}e^{-n(x+y)}\frac{J_0'(\sqrt{y(y+2x)})}{\sqrt{y(y+2x)}}x^2dydx, \\
	\mathcal{M}(n):=\int_0^{\infty}\int_0^{\infty}e^{-n(x+y)}\frac{I_0'(\sqrt{y(y+2x)})}{\sqrt{y(y+2x)}}x^2dydx.
	\end{matrix}
\end{equation}
Since $I_0$ grows exponentially, $\mathcal{M}$ is divergent for $n\leq1$. For $n>1$ they are
\begin{equation}
	\mathcal{N}(n)=\frac{1}{n^2+1}, \; \mathcal{M}(n)=\frac{1}{n^2-1}.
\end{equation}
The divergence could be interpreted as an infinite quantity of apparent dark matter, in the constant coefficient case, due to the expansion law near the first instant $\tau_I=-\infty$. However, any known physical theory fails near the Big Bang, so we should put a cut-off on integrals $\mathcal{N},\mathcal{M}$ that makes them finite.

For example, a natural cut off (Observation \ref{cut off}) for our theory is $\tau_{min}$ such that
\begin{equation}
\label{min}
	a(\tau_{min})^{1-n}=\left(\frac{1}{n^2}-\frac{1}{n^2+1}\right)\frac{8\pi G}{H_0^2}\langle\tilde{\rho}_0\rangle.
\end{equation}
As we will see (Observation \ref{pole vanish}), the pole in $n=1$ will be canceled by a factor $(n-1)$, so we can extend the $\mathcal{N}, \mathcal{M}$ functions also for $n<1$, which are the physical values. What we perform is essentially  a renormalization via analytic continuation.

In the constant coefficients case, the relation (\ref{alg rel}) can be expressed with an "effective density"
\begin{equation}
	\Omega_{eff}:=\Omega_{M0}-2\Omega_{\Lambda 0}=-2\frac{\ddot{a}_D(t_0)}{H_{D0}^2}.
\end{equation}
\begin{oss}
	The effective density results to be proportional to the deceleration parameter: $\Omega_{eff}=2q_0$. This means it may be negative. In general, $\Omega_{eff}\in[-2; 1]$.

	Pay attention: this is not a violation of the weak energy condition. The true matter-energy density is $\overline{\rho}+\tilde{\rho}$, which is always positive. This $\rho_{eff}:=\frac{3H_{D0}^2}{8\pi G}\Omega_{eff}$ is only a fictitious density, without physical existence.
\end{oss}
To obtain an explicit expression for the effective density, we will use the following conventions
\begin{center}
	$K(n):=8\pi\left(\frac{1}{3}\mathcal{N}-\frac{1}{n^2}\right)G=8\pi\left(\frac{1}{3(n^2+1)}-\frac{1}{n^2}\right)G$; $K'(n):=8\pi\left(\mathcal{M}-\frac{1}{n^2}\right)G=8\pi\left(\frac{1}{(n^2-1)}-\frac{1}{n^2}\right)G$.
\end{center}
Moreover, we use $\overline{t}$ as most suitable variable, with $\overline{t}_0:=\overline{t}(t_0)$.
\begin{thm}
	If on a spatially flat metric dominated by a $w=-\frac{1}{3}$ dark matter-kind, we put an inhomogeneity of matter $\tilde{\rho}$, the present deceleration parameter of the averaged perturbed metric can be interpreted with an amount of effective density
	\begin{multline}
	\label{rho eff}
		\rho_{eff}(\langle\tilde{\rho}_0\rangle; H_{D0}; n)=\frac{3(1-n)}{16\pi G}H_0^4\overline{t}_0^2\frac{H_0^2\overline{t}_0^2-1}{[(1+K/K')H_0^2\overline{t}_0^2-K/K']^2}\cdot \\
		\cdot\left[\left((3-n)K/K'+(7-n)\right)H_0^2\overline{t}_0^2-\left((5-n)K/K'-(n+3)\right)\right],
	\end{multline}
	such that
	\begin{equation}
	\label{alg sys}
		\begin{cases}
			K'H_0^{n-1}\langle\tilde{\rho}_0\rangle\overline{t}_0^{n+1}+1=H_0^2\overline{t}_0^2 \\
			H_{D0}=\frac{H_0}{2}\frac{(1-n)H_0^2\overline{t}_0^2+(1+n)}{\sqrt{(1+K/K')H_0^2\overline{t}_0^2-K/K' }}
		\end{cases}.
	\end{equation}
\end{thm}
\begin{proof}
	Remember

	$\bold{a}(\overline{t})=H_0\overline{t}\sqrt{1-K'H_0^{n-3}\langle\tilde{\rho}_0\rangle\overline{t}^{n-1}}$.

	We can derive it exploiting

	$\frac{dt}{d\overline{t}}=\sqrt{1+KH_0^{n-3}\langle\tilde{\rho}_0\rangle\overline{t}^{n-1}}$,

	and substituting inside the first two equations of (\ref{dark sys}) we get (\ref{alg sys}). Remember the known parameters are $\langle\tilde{\rho}_0\rangle, H_{D0}, n$, so this system determines $\overline{t}_0, H_0$. Deriving $\textbf{a}$ again, we get (\ref{rho eff}).
\end{proof}
This gives the quantity of fictitious matter and cosmological constant an observer would need to justify the measured distortion of deceleration parameter, if there is an average inhomogeneity of matter $\langle\tilde{\rho}_0\rangle$. Substituting the values for $\langle\tilde{\rho}_0\rangle$, it's possible to evaluate the magnitude of these effects.
\begin{oss}
\label{pole vanish}
	The factor $(1-n)$ in the $\rho_{eff}$ compensates for the pole at $n=1$ inside $K'(n)$. This justifies the renormalization we performed for the $\mathcal{M},\mathcal{N}$ integrals.
\end{oss}

\section{Evaluation of dark matter and dark energy magnitude}

\subsection{Numerical values}

Since we assumed inhomogeneities are small, so that we can use the perturbative approach, here it is possible to expand all the theorem's quantities at the first order in $\langle\tilde{\rho}_0\rangle$. From (\ref{alg sys}) we get
\begin{align}
	\overline{t}_0=\frac{1}{H_{D0}}+\frac{K+(n+1)K'}{2H_{D0}^3}\langle\tilde{\rho}_0\rangle+o(\langle\tilde{\rho}_0\rangle), \cr
	H_0=H_{D0}+\frac{K+nK'}{2H_{D0}}\langle\tilde{\rho}_0\rangle+o(\langle\tilde{\rho}_0\rangle).
\end{align}
The factor $(H_0^2\overline{t}_0^2-1)$ inside $\rho_{eff}$ has no zeroth order term, so we must take only the zeroth order term for all the other factors, which simplifies the calculation of
\begin{equation}
	\rho_{eff}=3(1-n)\left(-\frac{2-n}{n^2-1}-\frac{1}{3(n^2+1)}-\frac{1-n}{n^2}\right)\langle\tilde{\rho}_0\rangle+o(\langle\tilde{\rho}_0\rangle).
\end{equation}
We see clearly here how the pole $n=1$ vanishes. We can express our result as
\begin{equation}
\label{formula}
	\frac{\Omega_{M0}-2\Omega_{\Lambda 0}}{\Omega_{IM0}}\cong ract(n):=3\frac{2-n}{n+1}-\frac{1-n}{n^2+1}-\frac{(1-n)^2}{n^2}
\end{equation}
where we called $\Omega_{IM0}:=\frac{8\pi G}{3H_{D0}^2}\langle\tilde{\rho}_0\rangle$ the quantity of inhomogeneous matter. Except for matter and cosmological constant, which are almost fictitious, the only component is the dark energy $\Omega_{E0}:=\Omega_{w0}|_{w=-\frac{1}{3}}$. Remembering $\overline{\Omega}_{E0}=1$, as it was dominant in the background expansion, we have at first order
\begin{equation}
	\Omega_{E0}=\left(\frac{H_0}{H_{D0}}\right)^2=1+\frac{K+nK'}{H_{D0}^2}\langle\tilde{\rho}_0\rangle+o(\langle\tilde{\rho}_0\rangle).
\end{equation}
Since the sum of all $\Omega_{w0}$ is always $1$, we get
\begin{equation}
	\frac{\Omega_{M0}+\Omega_{\Lambda0}}{\Omega_{IM0}}\cong sum(n):=\frac{n+1}{n^2}+\frac{n}{1-n^2}-\frac{1}{3}\frac{1}{n^2+1}.
\end{equation}
This allows us to obtain $\Omega_{\Lambda,M0}$ as function of the perturbation $\Omega_{IM0}$. As we didn't put any cosmological constant in the true universe, and we put the matter only as the inhomogeneous one,
\begin{equation}
	\begin{matrix}
		\Omega_{F\Lambda 0}=\Omega_{\Lambda 0}\cong\frac{sum(n)-ract(n)}{3}\Omega_{IM0} \\
		\Omega_{M0}\cong\frac{2sum(n)+ract(n)}{3}\Omega_{IM0} \\
		\Omega_{FM0}=\Omega_{M0}-\Omega_{IM0}\cong\left(\frac{2sum(n)+ract(n)}{3}-1\right)\Omega_{IM0}
	\end{matrix}
\end{equation}
are the fictitious matter and cosmological constant.  Remembering the condition (\ref{cond n}), we can try to replace $n\cong\frac{2}{3}$; this gives
\begin{equation}
\label{meas}
	\begin{matrix}
		ract(\frac{2}{3})\cong-3.38; \; sum(\frac{2}{3})\cong4.72; \\
		\frac{\Omega_{\Lambda 0}}{\Omega_{IM0}}\cong\frac{27}{10}; \; \frac{\Omega_{M0}}{\Omega_{IM0}}\cong2.02; \; \frac{\Omega_{FM0}}{\Omega_{IM0}}\cong1.02 \; .
	\end{matrix}
\end{equation}
In particular, we obtain $\Omega_{M0}>\Omega_{IM0}$ and $\Omega_{FM0}>0$, so the magnification effect is verified (Observation \ref{magn}).

We can compare the formulas (\ref{meas}) with the most recent measures of cosmological parameters in our universe \cite{measures}.
\begin{center}
	$\Omega_{B0}\cong0.043\pm0.004$; $\Omega_{M0}\cong0.315\pm0.007$; $\Omega_{\Lambda 0}\cong0.685\pm0.007 \; .$
\end{center}
What we find is
\begin{equation}
	\Omega_{IM0}\cong0.156; \; \Omega_{FM0}\cong0.159; \; \Omega_{TDM0}\cong0.113; \; \Omega_{F\Lambda 0}\cong0.421 \; .
\end{equation}
The fraction of actually existing dark matter would be $\frac{\Omega_{TDM0}}{\Omega_{IM0}}\cong72.4\%$, and not the $\frac{\Omega_{M0}-\Omega_{B0}}{\Omega_{M0}}\cong86.3\%$ as is usually believed. Moreover, the most part of the cosmological constant would be fictitious, so that the only quantity which cannot be brought back to inhomogeneity effects is $\Omega_{\Lambda 0}-\Omega_{F\Lambda 0}\cong0.264 \; .$

Remember that all these values are provisional. The quantitative results could change in a model with non constantly expanding background. Comparing Figure 1 with Figure 2 and 3, we can imagine that the inhomogeneity effects could be stronger under a dominance of radiation or matter, since the "echos" result to develop faster (the same shape to get which under constant expansion it needs $\tau-\tau'=50 Gy$, is reached under matter in $20 Gy$ and under radiation in $13 Gy$). The real universe passed a phase of radiation dominance and then of matter dominance, so we can expect higher values for $ract(n), sum(n)$ and a more complete explanation of the dark matter and the cosmological constant.

\subsection{Do we find an inflation-like effect?}

If we are able to evaluate the fictitious quantity of cosmological constant, it would be interesting to obtain its variation during time and check if in the past was bigger. It would provide an explanation for the inflationary theory. So we fix a past instant $t_1$. We put our observer at this time and we wander how much inhomogeneity effect he sees. As in previous calculations (\ref{dark fried}), the observer considers a purely homogeneous model\footnote{Coherently with the previous notation, we write $Q_1:=Q(t_1)$ for any quantity.}
\begin{equation}
\label{dark fried1}
	\left(\frac{\dot{a}_D}{H_{D1}}\right)^2=\sum_w\Omega_{w1}a_D^{-3w-1}.
\end{equation}
Since he lives in $t_1$, its effective expansion parameter is fixed as $a_D(t_1)=1$. This means it is reduced by a factor $\textbf{a}_1:=\textbf{a}(t_1)$ with respect to $\textbf{a}(t)$. The setting of parameters (\ref{dark sys}) become
\begin{equation}
	\begin{cases}
		a_D(t_1):=1=\frac{\bold{a}(t_1)}{\textbf{a}_1} \\
		\dot{a}_D(t_1)=H_{D1}:=\frac{\dot{\bold{a}}(t_1)}{\textbf{a}_1} \\
		\ddot{a}_D(t_1):=\frac{\ddot{\bold{a}}(t_1)}{\textbf{a}_1}
	\end{cases}.
\end{equation}
From them we have the analogous of (\ref{alg sys}) and (\ref{rho eff}):
\begin{multline}
\label{rho eff1}
	\rho_{eff}(\langle\tilde{\rho}_0\rangle; H_{D1}; n)=\frac{3(1-n)}{16\pi G}\frac{H_0^4\overline{t}_1^2}{\textbf{a}_1^4}\frac{H_0^2\overline{t}_1^2-\textbf{a}_1^2}{[(1+K/K')H_0^2\overline{t}_1^2-(K/K')\textbf{a}_1^2]^2}\cdot \\
	\cdot\left[\left((3-n)K/K'+(7-n)\right)H_0^2\overline{t}_1^2-\left((5-n)K/K'-(n+3)\right)\textbf{a}_1^2\right], \; s.t. \\
	\begin{cases}
		K'H_0^{n-1}\langle\tilde{\rho}_0\rangle\overline{t}_1^{n+1}+\textbf{a}_1^2=H_0^2\overline{t}_1^2 \\
		H_{D1}=\frac{H_0}{2}\frac{(1-n)\textbf{a}_1^{-2}H_0^2\overline{t}_1^2+(1+n)}{\sqrt{(1+K/K')H_0^2\overline{t}_1^2-(K/K')\textbf{a}_1^2 }}
	\end{cases}.
\end{multline}
\begin{oss}
	Notice that setting $t_1=t_0$, $\textbf{a}_1=1$ and $H_{D1}=H_{D0}$, and we turn back to (\ref{alg sys}) and (\ref{rho eff}).
\end{oss}
We expand again the quantities in (\ref{rho eff1}) at the first order in $\langle\tilde{\rho}_0\rangle$:
\begin{equation}
	\begin{matrix}
		 H_0=H_{D1}a_1+\frac{K+nK'}{2H_{D1}}\textbf{a}_1^{n-2}\langle\tilde{\rho}_0\rangle+o(\langle\tilde{\rho}_0\rangle); \\
		\Omega_{M1}-2\Omega_{\Lambda 1}=ract(n)\textbf{a}_1^{n-3}\Omega_{IM1}+o(\Omega_{IM1}).
	\end{matrix}
\end{equation}
It is also
\begin{equation}
	\Omega_{E1}=\left(\frac{H_0}{H_{D1}\textbf{a}_1}\right)^2=1+\frac{K+nK'}{H_{D1}^2}\textbf{a}_1^{n-3}\langle\tilde{\rho}_0\rangle+o(\langle\tilde{\rho}_0\rangle)
\end{equation}
from which
\begin{equation}
	\Omega_{M1}+\Omega_{\Lambda1}=sum(n)\textbf{a}_1^{n-3}\Omega_{IM1}+o(\Omega_{IM1}).
\end{equation}
We get the past fictitious matter and cosmological constant as functions of $\Omega_{IM1}$:
\begin{equation}
\label{num1}
	\begin{matrix}
		\Omega_{F\Lambda 1}=\Omega_{\Lambda 1}\cong\frac{sum(n)-ract(n)}{3}\textbf{a}_1^{n-3}\Omega_{IM1} \\
		\Omega_{M1}\cong\frac{2sum(n)+ract(n)}{3}\textbf{a}_1^{n-3}\Omega_{IM1} \\
		\Omega_{FM1}=\Omega_{M1}-\Omega_{IM1}\cong\left(\frac{2sum(n)+ract(n)}{3}-1\right)\textbf{a}_1^{n-3}\Omega_{IM1}.
	\end{matrix}
\end{equation}
However, we live in $t_0$ and we are not able to measure $\Omega_{IM1}$, as an observer in $t_1$ does. We should convert our formulas with
\begin{equation}
	\Omega_{IM1}=\left(\frac{H_{D0}}{H_{D1}}\right)^2\textbf{a}_1^{n-3}\Omega_{IM0}.
\end{equation}
The factor $\Omega_{IM0}$ is yet at the first order of the perturbation, so we must evaluate the others at zeroth order, for which

$H_{D0}=H_0+O(\langle\tilde{\rho}_0\rangle)=H_{D1}\textbf{a}_1+O(\langle\tilde{\rho}_0\rangle)$,

so (\ref{num1}) become
\begin{equation}
	\begin{matrix}
		\Omega_{F\Lambda 1}=\Omega_{\Lambda 1}\cong\frac{sum(n)-ract(n)}{3}\textbf{a}_1^{2n-4}\Omega_{IM0} \\
		\Omega_{M1}\cong\frac{2sum(n)+ract(n)}{3}\textbf{a}_1^{2n-4}\Omega_{IM0} \\
		\Omega_{FM1}\cong\left(\frac{2sum(n)+ract(n)}{3}-1\right)\textbf{a}_1^{2n-4}\Omega_{IM0}
	\end{matrix}
\end{equation}
Evaluating again $n\cong\frac{2}{3}$, the time dependence of the cosmological constant is
\begin{equation}
	\Omega_{F\Lambda 1}\cong\frac{27}{10}\textbf{a}_1^{-\frac{8}{3}}\Omega_{IM0}
\end{equation}
So $\Omega_{\Lambda}$ was pretty greater in the past, which can be interpreted as an inflationary epoch. It was even too much large, according to these equations. Assuming $\Omega_{IM0}=\Omega_{B0}\cong0.043\pm0.004$ as we see, the cosmological constant reach the maximum physical value $\Omega_{F\Lambda 1}=1$ at
\begin{equation}
	\textbf{a}_1=\textbf{a}_{mmin}:=\left(\frac{27}{8}\Omega_{IM0}\right)^{\frac{3}{8}}\cong0.446 \; .
\end{equation}
Before this instant, our equations for $\Omega_{F\Lambda 1}$ are no more valid. We can heuristically state that in the previous epoch the cosmological constant is completely dominant $\Omega_{F\Lambda 1}=1$. This until we reach the natural cut off (\ref{min}). With our parameters, it is
\begin{equation}
	\textbf{a}_{min}=\left[3\left(\frac{1}{n^2}-\frac{1}{n^2+1}\right)\Omega_{IM0}\right]^{\frac{1}{1-n}}\cong(4.67\Omega_{IM0})^3\cong0.00811 \; .
\end{equation}
We can summarize our knowledge about the cosmological constant variations as
\begin{equation}
	\Omega_{F\Lambda}(t)\cong\begin{cases}
		unknown \; for \; \textbf{a}(t)<8.11\cdot10^{-3} \\
		1 \; for \; 8.11\cdot10^{-3}<\textbf{a}(t)<0.446 \\
		0.116\cdot \textbf{a}(t)^{-\frac{8}{3}} \; for \; a(t)>0.446
	\end{cases}.
\end{equation}
We can interpret the second epoch as inflation.

\section{Conclusions}

In this work we revalued the train of thoughts of \cite{sergio}, beyond its lacks, developing it in a truly general relativistic context and for a completely generic distribution of inhomogeneities. We got that the real amount of the dark matter content in our
universe is lower than predicted by the Cosmological Standard Model, due to retarded relativistic effects, as we expected. Performing the formalism, we were lead also to consider the possibility of a similar correction also for the cosmological constant and, observing its variations, even an explanation for the inflationary scenario arose.

Our calculation is approximate in more senses, so for reliable numerical values a more deep formalism is required. Above all, it is required to get gravitational retarded potential also for not constantly expanding universe, as the real case is. However, as we showed, it is reasonable to believe that in the real case the same effects would arise, with even greater magnitude, which could explain the most part of the dark matter and cosmological constant usually required. Another possible deepening could be the consideration of some spatial curvature and its consequences. All these complications will need numerical algorithms to be performed.

A possible generalization of the present work could be also an evaluation of the gravitational forces due to retarded potentials $\tilde{g}_{\mu\nu}$. The comparison to the attraction of dark matter in galaxies or cluster could justify it completely or partially. To obtain a theoretical prevision of the dark matter distribution, mathematical tools to describe the baryonic matter distribution would be necessary. Since the matter inhomogeneity seems to have a fractal shape \cite{fract}, it could be useful a study of singular distributions as sources in General Relativity, and their application to cosmology.

The study of solutions of (\ref{PDE}) without source could be useful for the interpretation of ancient gravitational waves, since universe has expanded from their emission.

\medskip
\textbf{Acknowledgments}. We thank Oliver Piattella, Mariateresa Crosta and Marco Giammaria for useful discussions.

\end{document}